\documentclass[11pt]{article}

\usepackage[a4paper, margin=2cm]{geometry}

\usepackage[T1]{fontenc}

\usepackage{graphicx}
\usepackage{color}
\usepackage[bookmarksdepth=2]{hyperref}

\usepackage{amsmath}
\usepackage{amssymb}

\usepackage{xspace} 
\usepackage{xcolor}
\usepackage{algorithm}
\usepackage[noend]{algpseudocode}

\usepackage{authblk}

\usepackage{lineno}

\usepackage[draft]{changes}
\definecolor{darkgreen}{rgb}{0,0.7,0}
\definechangesauthor[name={reviewer}, color=darkgreen]{rev}
\definecolor{darkred}{rgb}{0.8,0,0}
\definechangesauthor[name={yongho}, color=darkred]{yshin}
\setcommentmarkup{\textcolor{authorcolor}{[#1]}}

\usepackage{amsthm}
\theoremstyle{plain}
\newtheorem{theorem}{Theorem}
\newtheorem{proposition}{Proposition}
\newtheorem{lemma}{Lemma}

\theoremstyle{definition}

\theoremstyle{remark}

\usepackage{tikz}
\usetikzlibrary {arrows.meta}

\newcommand{\calM}{\mathcal{M}}

\newcommand{\I}{\mathbb{I}}
\newcommand{\R}{\mathbb{R}}

\newcommand{\xstar}{x^\star}
\newcommand{\ystar}{y^\star}
\newcommand{\Mstar}{M^\star}

\newcommand{\xtilde}{\Tilde{x}}
\newcommand{\ytilde}{\Tilde{y}}
\newcommand{\ztilde}{\Tilde{z}}
\newcommand{\Gtilde}{\Tilde{G}}
\newcommand{\Mtilde}{\Tilde{M}}
\newcommand{\gammatilde}{\Tilde{\gamma}}

\newcommand{\fl}{\textsc{Facility Location}\xspace}
\newcommand{\ufl}{\textsc{Uncapacitated Facility Location}\xspace}
\newcommand{\flm}{\textsc{Facility Location with Matching}\xspace}
\newcommand{\ecfl}{\textsc{Even-constrained Facility Location}\xspace}
\newcommand{\FLM}{\textsc{FLM}\xspace}
\newcommand{\UFL}{\textsc{UFL}\xspace}
\newcommand{\ECFL}{\textsc{ECFL}\xspace}

\newcommand{\Reroute}{\textsc{Reroute}\xspace}
\newcommand{\open}{\mathrm{open}}
\newcommand{\connflm}{\mathrm{conn}_\textsc{FLM}}
\newcommand{\connufl}{\mathrm{conn}_\textsc{UFL}}
\newcommand{\PMM}{P_\mathrm{MM}}

\newcommand{\transfer}{\tau^i_{e' \to e}}
\newcommand{\rhoUFL}{\rho_\textsc{UFL}}

\DeclareMathOperator{\E}{\mathbb{E}}
\DeclareMathOperator{\argmin}{argmin}

\DeclareMathOperator{\conv}{conv}

\title{Servicing Matched Client Pairs with Facilities}
\date{}

\author[1]{Fateme Abbasi}
\author[1]{Martin Böhm}
\author[1]{Jarosław Byrka}
\author[2]{Matin Mohammadi}
\author[1]{Yongho Shin}

\affil[1]{Institute of Computer Science, University of Wrocław, Poland}
\affil[2]{Department of Computer Engineering, Sharif University of Technology, Iran}

\begin{document}

\maketitle             
\begin{abstract}
We study \flm, a \textsc{Facility Location} problem where, given additional information about which pair of clients is compatible to be matched, we need to match as many clients as possible and assign each matched client pair to a same open facility at minimum total cost.
The problem is motivated by match-making services relevant in, for example, video games or social apps. It naturally generalizes two prominent combinatorial optimization problems---\ufl and \textsc{Minimum-cost Maximum Matching} (in metric). \flm also generalizes the \ecfl problem studied by Kim, Shin, and An (Algorithmica 2023).

We propose a linear programming (LP) relaxation for this problem, and present a $3.868$-approximation algorithm. Our algorithm leverages the work on bifactor-approximation algorithms (Byrka and Aardal, SICOMP 2012); our main technical contribution is a rerouting subroutine that reroutes a fractional solution to be supported on a fixed maximum matching with only small additional cost.
For a special case where all clients are matched, we provide a refined algorithm achieving an approximation ratio of $2.218$. As our algorithms are based on rounding an optimal solution to the LP relaxation, these approximation results also give the same upper bounds on the integrality gap of the relaxation.
\end{abstract}
%
\section{Introduction} \label{sec:intro}
\fl is a quintessential problem in combinatorial optimization with wide-ranging applications in practice~\cite{LocationScience}.
In this problem, given a set of candidate facility locations and a set of clients, we need to open facilities at some locations and assign every client to an open facility, where the objective is to minimize the total cost incurred by opening facilities and assigning all clients.
Due to its high relevance in practice, the problem has been extensively studied under various settings ranging from classical \ufl (\UFL)~\cite{hochbaum1982approximation,shmoys1997approximation,chudak1998improved,guha1999greedy,jain2002new,sviridenko2002improved,charikar2005improved,mahdian2006approximation,byrkaaardal2010,li2013}
to ones that incorporate additional constraints embodying plausible scenarios;
to give a very limited list of these extensions, it has been studied under 
capacitated/lower-bounded settings~\cite{korupolu2000analysis,chudak1999improved,levi2004lp,aggarwal20103,bansal20125,an2017lp,li2019facility},
robust settings~\cite{charikar2001algorithms,friggstad2019approximation,dabas2024capacitated},
dynamic/online settings~\cite{meyerson2001online,fotakis2008competitive,an2017dynamic},
and submodular settings~\cite{svitkina2010facility,abbasi2024loglog}.

In this paper, we study a \fl problem originating from match-making services present in many applications such as one-on-one real-time video games (e.g.,~\cite{sf6,tekken8,hearthstone}) and online social/dating apps (e.g.,~\cite{linkupfitness,tinder}).
For a concrete example, \emph{Street Fighter 6}~\cite{sf6} is a popular one-on-one real-time fighting game that provides an online match-making service for tens of thousands users playing at the same time~\cite{sf6stats}.
In such applications, the service provider needs to open servers and assign matched pairs of users to the open servers, and there are several desiderata that the service should satisfy.
As user retention is an essential aspect for the service, the service provider would like to make as many matches as possible.
However, a user may not be compatible to be matched with any arbitrary user; for example in a video game, two users would be incompatible to be matched because their skill levels are discrepant, or they have been recently matched together.
Moreover, it is essential to assign each matched pair to a nearby open server to maintain a connection with imperceptible latency for both players.

We formulate a \fl problem, called \flm (\FLM), that succinctly captures this scenario as follows.
In this problem, we are additionally given information of which client pairs are compatible to be matched.
This information can be modeled by a simple graph, called a \emph{compatibility graph}, on the client set where each edge represents the compatibility of being matched between the two endpoint clients.
We need to find a \emph{maximum matching} in the compatibility graph and assign both clients of each matched pair to the same open facility.
The objective is to minimize the total cost incurred by opening the facilities and assigning the matched clients.

It is noteworthy that this problem naturally generalizes two prominent combinatorial optimization problems --- \ufl and \textsc{Minimum-cost Maximum Matching} (in metric).
As \UFL is well-known to be NP-hard, we aim at designing approximation algorithms for \FLM in this paper.
Moreover, besides these classical problems, we also find that the problem fits neatly into the existing work on \ecfl (\ECFL), studied by Kim, Shin, and An~\cite{kim2023constant}, where each open facility is required to service exactly even number of clients.
They presented a $2 \, \rhoUFL$-approximation algorithm for \ECFL through a reduction to a $\rhoUFL$-approximation algorithm for \UFL.
Observe that \FLM generalizes \ECFL because it is equivalent to an instance where the number of clients is even and every pair of clients is compatible to be matched, i.e., a complete graph is given as the compatibility graph.

One challenge in design of approximation algorithms for \FLM is to obtain a reasonable lower bound on the cost of an optimal solution.
In particular, as was also noticed by Kim et al.~\cite{kim2023constant}, 
the cost of an optimal \UFL solution without the matching constraint cannot serve as a reasonable lower bound.
To see this fact, imagine an instance with two candidate locations far away from each other while there exists one client at each location, respectively.
If the opening costs of the candidate locations are negligible to the distance between them, we can see that any \FLM solution must pay the distance between the clients whereas an optimal \UFL solution does not.

This negative result also implies that the standard linear programming (LP) relaxation for \UFL has an unbounded integrality gap for \FLM.
This is rather unsatisfactory because an LP relaxation with a bounded integrality gap often serves as a key component in design of approximation algorithms for \fl and beyond~\cite{williamson2011design}.
In fact, Kim et al.~\cite{kim2023constant} also posed an open question whether \ECFL, among other problems studied in that paper, admits an algorithmically useful LP relaxation with a bounded integrality gap.

\paragraph{Our contribution.}
The main contribution of this paper is in presenting efficient approximation algorithms for \FLM.
To this end, we first formulate in Section~\ref{sec:lp} an LP relaxation attained by combining the standard LP relaxation for \UFL with the matching polytope~\cite{edmonds1965maximum}.
Our algorithmic results to be presented shortly are based on an LP rounding technique which yields the upper bounds on the integrality gap of the LP relaxation same as the approximation ratios.
This answers in the affirmative the aforementioned open question raised by Kim et al.~\cite{kim2023constant}.

In Section~\ref{sec:maxmat}, we present a $3.868$-approximation algorithm for \FLM.
The algorithm initially solves the LP relaxation and computes a minimum-cost maximum matching.
It then reroutes the assignment of the optimal LP solution so as to construct a (fractional) solution supported by the matching at small additional cost.
We can then interpret this modified solution as a feasible solution to the standard LP relaxation for \UFL where each matched edge is regarded as a ``meta-client''.
Hence, the algorithm runs an approximation algorithm for \UFL to obtain the final integral solution.
We remark that the rerouting subroutine incurs different factors on the opening and assignment costs, leading us to use a \emph{bifactor-approximation algorithm}~\cite{byrkaaardal2010} to achieve the claimed approximation ratio.
Moreover, if the compatibility graph is perfectly matchable (i.e., there exists a perfect matching in the graph), we can enhance the rerouting step by exploiting a property of perfect matchings, resulting in a $2.373$-approximation algorithm.

We further investigate the case where the compatibility graph is perfectly matchable in Section~\ref{sec:perfmat}, and present a tailored algorithm that achieves an improved approximation ratio of $2.218$.
Unlike the previous algorithm which computes a maximum matching upfront and reroutes the assignment to construct a solution supported by this matching, this algorithm regards the optimal solution to our LP relaxation for \FLM as a feasible solution to the standard LP relaxation for \UFL.
It then runs a bifactor-approximation algorithm for \UFL to obtain open facilities, and computes a best perfect matching with respect to these obtained open facilities.
Here we highlight that, since \ECFL is a special case of \FLM where the compatibility graph is perfectly matchable, our algorithm yields an improvement over their $2 \, \rhoUFL \approx 2.976$ to $2.218$ in approximation ratio for \ECFL, where $\rhoUFL \approx 1.488$ is due to Li~\cite{li2013}.

We conclude this paper by discussing inherent limitations of our approach and plausible future directions in Section~\ref{sec:concl}.

\paragraph{Further related work.}
We list a few more problems that share in part the common theme with \FLM.
Gourdin, Labb{\'e}, and Laporte~\cite{gourdin2000uncapacitated} study a problem where the clients are collected along routes starting and ending at open facilities while each route can contain at most two clients.
Another example is by Candogan and Feng~\cite{candogan2024mobility}, who study a problem where each client is represented by two locations, and it suffices to connect one of its two locations to an open facility.
\section{Preliminaries} \label{sec:prelim}
\paragraph{Problem definition.}
We now formally define the \flm problem.
An instance of this problem can be represented by $(F, V, E, f, d)$ defined as follows.
We are given a set $F$ of facilities and a set $V$ of clients together lying on a metric $d : \binom{F \cup V}{2} \to \R_+$ satisfying the triangle inequality, i.e., for any $i, j, k \in F \cup V$, we have
$
    \textstyle d(i, j) \leq d(i, k) + d(k, j).
$
Each facility $i \in F$ has an opening cost $f(i) \geq 0$.
For the clients, we are additionally given a set $E \subseteq \binom{V}{2}$ of client pairs, called \emph{compatible pairs}, indicating the two clients are allowed to be matched together.
We call the graph $(V, E)$ the \emph{compatibility graph}.

For every $i \in F$ and $e = \{j, k\} \in E$, we define
$
    d(i, e) := d(i, j) + d(i, k)
$
for simplicity of presentation.
On the other hand, we write $d(e) := d(j, k)$.
Due to the triangle inequality, it is easy to observe the following:
\begin{lemma} \label{lem:prelim:distpairbound}
    For any $i \in F$ and $e \in E$, $d(e) \leq d(i, e)$.
\end{lemma}

A feasible solution to \FLM is defined by $(S, M, \sigma)$ where $S \subseteq F$ denotes the set of open facilities, $M$ denotes a maximum matching in the compatibility graph $(V,E)$, and $\sigma : M \to S$ denotes the assignment of each matched pair to an open facility.
The objective is to minimize the total cost incorporating the total opening cost and the total assignment cost, i.e.,
$
    \textstyle \sum_{i \in S} f(i) + \sum_{e \in M} d(\sigma(e), e).
$

\paragraph{\ufl.}
In the design of our approximation algorithms for \FLM, it is essential to understand the classic \ufl problem.
Recall that, in this problem, we are given $(F, C, f, d)$ consisting of a set $F$ of facilities, a set $C$ of clients, an opening cost $f : F \to \R_+$, and an assignment cost $d : F \times C \to \R_+$ satisfying the \emph{three-hop inequality},\footnote{This inequality is usually called the \emph{triangle inequality} in the literature, but we use this definition to distinguish it from the triangle inequality of $d$ in \FLM. In fact, one can easily construct from $d$ on $F \times C$ satisfying the three-hop inequality to an equivalent metric on $\binom{F \cup C}{2}$ satisfying the triangle inequality, but we stick to this definition since it is sufficient in our paper.}
i.e., for any $i, i' \in F$ and $j, j' \in C$, we have
$
    \textstyle d(i, j) \leq d(i, j') + d(i', j') + d(i', j).
$
The objective is to find $(S, \sigma)$, where $S$ denotes the set of open facilities and $\sigma : C \to S$ denotes the assignment of each client to an open facility, at minimum total cost
$
    \textstyle \sum_{i \in S} f(i) + \sum_{j \in C} d(\sigma(j), j).
$
We remark that $\sigma(j) = \argmin_{i \in S} d(i, j)$ for every $j \in C$ in \UFL.

Note that there exists a straightforward approximation-preserving reduction from \UFL to \FLM obtained by, given an \UFL instance, making one more copy of each client at the same location where each client is compatible with only its copy, respectively, and doubling the opening costs.
Therefore, \FLM inherits all the hardness/inapproximability results from \UFL.
\begin{theorem}[\cite{guha1999greedy}]
    \FLM is MAXSNP-hard.
    Moreover, there does not exist an approximation algorithm for \FLM with an approximation ratio better than $1.463$ unless $\textrm{NP} \subseteq \textrm{DTIME}(n^{\log \log n})$.
\end{theorem}

Recall that the standard LP relaxation for \UFL is formulated as follows:
\begin{align}
\min \quad & \textstyle \sum_{i \in F} f(i) \, y_i + \sum_{i \in F} \sum_{j \in C} d(i,j) \, x_{i,j} \tag{$\mathrm{LP}_{\UFL}$} \label{lp:ufl} \\
\text{s.t.} \quad & \textstyle \sum_{i \in F} x_{i,j} = 1, && \forall j \in C, \label{lp:ufl:assign} \\
& x_{i,j} \le y_i, && \forall i \in F,\, j \in C, \label{lp:ufl:open} \\
& x_{i,j} \ge 0,\, y_i \ge 0, && \forall i \in F,\, j \in C, \label{lp:ufl:nonneg}
\end{align}
where $y_i = 1$ indicates that facility $i$ is open, and $x_{i,j} = 1$ indicates that client $j$ is assigned to facility $i$.
Given a feasible solution $(x, y)$ to \eqref{lp:ufl}, let us denote by $\open(y)$ and $\connufl(x)$ the total opening and assignment costs incurred by $(x, y)$, respectively, i.e., $\open(y) := \sum_{i \in F} f(i) y_i$ and $\connufl(x) := \sum_{i \in F}\sum_{j \in C} d(i, j) x_{i, j}.$

We now define a \emph{bifactor-approximation LP-rounding algorithm}.
For some $\lambda_f$ and $\lambda_c$, we say an algorithm is a \emph{$(\lambda_f, \lambda_c)$-approximation LP-rounding algorithm} for \UFL if, given any feasible solution $(x, y)$ to \eqref{lp:ufl}, the algorithm outputs an integral feasible \UFL solution whose expected cost is bounded from above by $\lambda_f \cdot \open(y) + \lambda_c \cdot \connufl(x)$.
We will later exploit the bifactor-approximation LP-rounding algorithm guaranteed by the next proposition.
\begin{proposition}[\cite{byrkaaardal2010}] \label{prop:prelim:bifactor}
    For any $\lambda \geq 1.678$, there exists a $(\lambda, 1 + \frac{2}{e^{\lambda}})$-approximation LP-rounding algorithm for \UFL.
\end{proposition}

\paragraph{Maximum matching.}
Lastly, we present several known facts about maximum matchings that will be used in what follows.
Given two sets $A$ and $B$, we denote the symmetric difference between $A$ and $B$ by $A \mathbin{\triangle} B := (A \setminus B) \cup (B \setminus A)$.

\begin{proposition}
    Given a graph and two maximum matchings $M$ and $M'$ in the graph, $M \mathbin{\triangle} M'$ is a union of even-length alternating paths and cycles.
    If $M$ and $M'$ are perfect matchings, $M \mathbin{\triangle} M'$ is a union of even-length alternating cycles.
\end{proposition}

Given a graph $G = (V, E)$, for $U \subseteq V$,
we denote by $\delta(U) := \{ \{u, v\} \in E \mid |\{u, v\} \cap U| = 1 \}$ the set of edges across $U$, and by $E[U] := \{ \{u, v\} \in E \mid \{u, v\} \subseteq U\}$ the set of edges within $U$.
For $v \in V$, we simply write $\delta(v) := \delta(\{v\})$.
Let $\calM(G)$ denote the set of all maximum matchings in $G$, and $\nu(G)$ denote the size of a maximum matching in $G$. We may omit the parameter graph $G$ if clear from context.
We define
$
    \PMM(G) := \conv ( \{ \chi^M \mid M \in \calM \} ),
$
where $\chi^M \in \{0, 1\}^E$ denotes the characteristic
vector of $M \subseteq E$ while $\conv(T)$ represents the convex hull of a set $T$. 

Following are well-known polyhedral properties of $\PMM(G)$.
For the sake of completeness, we provide a proof for the first half of Proposition~\ref{prop:prelim:pmm} in Appendix~\ref{app:deferpf}; see, e.g., Schrijver~\cite{schrijver2003combinatorial} and Korte and Vygen~\cite{korte2008combinatorial} for the proofs of the other propositions.

\begin{proposition} \label{prop:prelim:pmm}
    $\PMM(G)$ is comprised of $z \in \R^E$ satisfying
   \begin{align}
        &
        \textstyle \sum_{e \in \delta(v)} z_e \leq 1, && \forall v \in V, \label{const:pmm:degree}
        \\&
        \textstyle \sum_{e \in E[U]} z_e \leq \frac{|U|-1}{2}, && \forall U \subseteq V : |U| \text{ is odd}, \label{const:pmm:oddset}
        \\&
        \textstyle \sum_{e \in E} z_e = \nu, \label{const:pmm:size}
        \\&
        z_e \geq 0, && \forall e \in E. \label{const:pmm:nonneg}
    \end{align}
    Moreover, if $G$ is perfectly matchable (i.e., $\nu = \frac{|V|}{2}$), $\PMM(G)$ is also determined by $z \in \R^E$ satisfying
    \begin{align}
        &
        \textstyle \sum_{e \in \delta(v)} z_e = 1, && \forall v \in V, \label{const:ppm:degree}
        \\&
        \textstyle \sum_{e \in \delta(U)} z_e \geq 1, && \forall U \subseteq V : |U| \text{ is odd}, 
        \label{const:ppm:oddset}
        \\&
        z_e \geq 0, && \forall e \in E. \label{const:ppm:nonneg}
    \end{align}
\end{proposition}
\begin{proposition} \label{prop:prelim:sep}
    There exists a polynomial-time separation oracle for $\PMM(G)$.
\end{proposition}
\begin{proposition}
\label{poly-mmp} \label{prop:prelim:concom}
    Given any feasible solution $z \in \PMM(G)$, one can in polynomial time decompose $z$ into a convex combination of $O(|E|)$ maximum matchings in $G$. That is, one can find a coefficient vector $\gamma \in \R^\calM_+$ with $O(|E|)$ nonzero entries satisfying $\sum_{M \in \calM} \gamma_M = 1$ and $z_e = \sum_{M \in \calM} \gamma_M \chi^M_e$ for all $e \in E$.
\end{proposition}
\section{LP relaxation}  \label{sec:lp}
In this section, we present the LP relaxation for \FLM that will be used in our approximation algorithms.
Let $G = (V, E)$ be the compatibility graph of the given \FLM instance.
We formulate an LP relaxation as follows:
\begin{align}
\min  \quad & \textstyle \sum_{i \in F} f(i) \, y_i + \sum_{i \in F} \sum_{e \in E} d(i, e) \, x_{i,e} \tag{$\textrm{LP}_\FLM$} \label{lp:flm}\\
\text{s.t.} \quad
& \textstyle x_e = \sum_{i \in F} x_{i,e}, && \forall e \in E, \nonumber \\
& \{x_e\}_{e \in E} \in \PMM(G), \label{lp:flm:pmm} \\
& \textstyle \sum_{e \in \delta(j)} x_{i,e} \le y_i, && \forall i \in F, j \in V, \label{lp:flm:flow} \\
& x_{i,e}, y_i \ge 0, && \forall i \in F, e \in E, \label{lp:flm:nonneg}
\end{align}
where $y_i = 1$ represents that facility $i \in F$ is open while $x_{i, e} = 1$ indicates that compatible pair $e \in E$ is together assigned to facility $i \in F$.

Similarly to \eqref{lp:ufl}, let $\connflm(x)$ be the total assignment cost incurred by a feasible solution $(x, y) \in \R^{F \times E}_+ \times \R^F_+$ to \eqref{lp:flm}, i.e.,
$
    \textstyle \connflm(x) := \sum_{i \in F} \sum_{e \in E} d(i, e) \, x_{i, e}.
$
Note that the opening cost of $(x, y)$ is denoted in the same way by $\open(y)$.
We also highlight that \eqref{lp:flm} can be solved in polynomial time because \eqref{lp:flm:pmm} is separable due to Proposition~\ref{prop:prelim:sep}.

 We argue that \eqref{lp:flm} is a feasible LP relaxation for \FLM.
 Indeed, for any feasible integral solution $(S, M, \sigma)$ to \FLM, let us consider the solution $(x, y)$ such that, for any $i \in F$ and $e \in E$,
 \[
     y_i := \I[i \in S]
     \text{ and }
     x_{i, e} := \I[e \in M \wedge i \in S \wedge \sigma(e) = i].
 \]
 Note that, if $e \in M$, there exists exactly one facility $i \in F$ such that $\sigma(e) = i$; otherwise, if $e \in E \setminus M$, $e$ is not assigned to any facility. 
 We can thus deduce $x_e = \I[e \in M]$ for every $e \in E$, implying that \eqref{lp:flm:pmm} is satisfied because $M$ is a maximum matching in $G$.
 Furthermore, if $j \in V$ is assigned to $i \in F$, there exists exactly one pair $e \in E$ such that $j \in e$ and $\sigma(e) = i$, showing that \eqref{lp:flm:flow} is also satisfied.
 It is trivial to see that \eqref{lp:flm:nonneg} is met.

We also discuss that \eqref{lp:flm:pmm} and \eqref{lp:flm:flow} are necessary to achieve a bounded integrality gap.
 Firstly, \eqref{lp:flm:pmm} is needed because \FLM includes the \textsc{Minimum-cost Perfect Matching} problem as a special case.
 For example, given a perfectly matchable compatibility graph $G$, one can imagine an LP where \eqref{lp:flm:pmm} is replaced with the degree bound constraints, i.e.,
 \[
     \textstyle \sum_{e \in \delta(j)} x_e = 1 \text{ for every $j \in V$}.
 \]
 However, consider an instance consisting of two free facilities $i_1, i_2$, located at different positions, and six mutually compatible clients $j_1, j_2, j_3, k_1, k_2, k_3$ such that $j_1, j_2, j_3$ are co-located with $i_1$ and $k_1, k_2, k_3$ are co-located with $i_2$.
 Observe that the cost of an optimal LP solution is 0 by setting
 \begin{align*}
     &
     y_{i_1} = y_{i_2} = 1 \text{ and }
     \textstyle x_{i_1, \{j_1, j_2\}} = x_{i_1, \{j_2, j_3\}} = x_{i_1, \{j_3, j_1\}} = x_{i_2, \{k_1, k_2\}} = x_{i_2, \{k_2, k_3\}} = x_{i_2, \{k_3, k_1\}} = \frac{1}{2}.
 \end{align*}
 On the other hand, the cost of any integral feasible solution should incur a positive cost since at least one client in $\{j_1, j_2, j_3\}$ should be matched with one in $\{k_1, k_2, k_3\}$, implying the unbounded integrality gap of this modified LP.

 Next, for \eqref{lp:flm:flow}, one can also imagine an LP where \eqref{lp:flm:flow} is substituted by a weaker set of constraints that
 \[
     x_{i, e} \leq y_i \text{ for every $i \in F$ and $e \in E$}.
 \]
 It is easy to see that this LP is also a feasible LP relaxation for \FLM since this set of constraints successfully models that a facility $i \in F$ should be open whenever there exists a pair $e \in E$ assigned to $i$.
 Consider however an instance of a single unit-cost facility $i$ and $n$ mutually compatible clients co-located altogether, where $n$ is a sufficiently large even number.
 The cost of any integral feasible solution should be 1 since we should integrally open the facility to serve the clients.
 On the other hand, setting $x_{i,e} = \frac{1}{n-1}$ for every $e \in E$ lets us fractionally open $i$ with $y_i = \frac{1}{n-1}$, incurring only $\frac{1}{n-1}$.
 This demonstrates the unbounded integrality gap of this weakened LP.
\section{3.868-approximation algorithm for \FLM} \label{sec:maxmat}
In this section, we present a $3.868$-approximation algorithm for \flm.
In Section~\ref{sec:maxmat:reroute}, we begin with describing a key technical subroutine called \Reroute that, given any feasible solution to \eqref{lp:flm} and a maximum matching $M$ in the compatibility graph, reroutes the (fractional) assignment to make it supported by only $M$ at a small additional cost while satisfying the feasibility to \eqref{lp:flm} with facility openings scaled by a factor of 2.
Equipped with this subroutine, we then present in Section~\ref{sec:maxmat:mainalg} our approximation algorithm for \FLM.
We also argue that, if the compatibility graph is perfectly matchable, our algorithm can be further improved to have an approximation ratio of $2.373$ by slightly modifying \Reroute.

\subsection{Subroutine \Reroute} \label{sec:maxmat:reroute}

\paragraph{Subroutine description.}
See Algorithm~\ref{alg:sub}
for a pseudocode of this subroutine. 
We are given a feasible solution $(x, y) \in \R^{F \times E}_+ \times \R^F_+$ to \eqref{lp:flm} and a maximum matching $M \subseteq E$ in the compatibility graph $G$.
Throughout the execution, let $\xtilde \in \R^{F \times E}_+$ be the solution maintained by the subroutine. 
We also maintain a decomposition of $\{\xtilde_e \}_{e \in E}$ into a convex combination of maximum matchings in $G$, i.e.,
$
    \textstyle \xtilde_e = \sum_{M' \in \calM} \gamma_{M'} \cdot \chi^{M'}_e
$
for every $e \in E$.
We initially have $\xtilde := x$; due to Proposition~\ref{prop:prelim:concom}, we can assume that $\gamma$ is supported by $O(|E|)$ maximum matchings at the very start of the execution.

The subroutine runs in iterations until $\gamma_M = 1$.
In each iteration where $\gamma_M < 1$ at the beginning, let $M' \neq M$ be any maximum matching other than $M$ in the support of $\gamma$, i.e., $\gamma_{M'} > 0$.
Let $P := e_1, \ldots, e_\ell$ be any (maximal) alternating path (or cycle) in $M \mathbin{\triangle} M'$ such that $\ell$ is even, every odd-indexed edge is in $M'$ (i.e., $\{e_1, e_3, \ldots, e_{\ell-1}\} \subseteq M'$), and every even-indexed edge is in $M$ (i.e., $\{ e_2, e_4, \ldots, e_\ell \} \subseteq M$).
For each $e' \in M' \cap P$, we choose a facility $i_{e'} \in F$ servicing a positive amount to $e'$ in $\xtilde$ (i.e., $\xtilde_{i_{e'}, e'} > 0$), and let $\varepsilon > 0$ be the minimum between $\gamma_{M'}$ and the minimum amount serviced by the chosen facilities among $M' \cap P$, i.e., 
$
    \textstyle \varepsilon := \min \{\gamma_{M'}, \min_{e' \in M' \cap P} \{\xtilde_{i_{e'}, e'} \}\}. 
$
We then reroute $\xtilde$ along $P$ as follows: for each $e' \in M' \cap P$, we decrease $\xtilde_{i_{e'}, e'}$ by $\varepsilon$ and increase $\xtilde_{i_{e'}, e}$ by $\varepsilon$, where $e \in M \cap P$ denotes the next edge to $e'$ in the order of $P$. See Fig.~\ref{fig:maxmat:reroute} for an illustration of this rerouting.
We also update $\gamma$ accordingly, i.e., $\gamma_{M'}$ gets decreased by $\varepsilon$ whereas $\gamma_{M' \mathbin{\triangle} P}$ gets increased by $\varepsilon$.
The final output of the subroutine is $(\xtilde, 2y)$.

\begin{algorithm}[t]
\caption{$\Reroute((x, y), M)$} \label{alg:sub}
\begin{algorithmic}[1]
    \State $\xtilde \gets x$
    \State Let $\gamma \in \R^\calM_+$ be the coefficient vector of a decomposition of $\{ \xtilde_e \}_{e \in E}$ into a convex combination of maximum matchings
    \While{$\gamma_M < 1$}
        \State Let $M' \neq M$ be any maximum matching with $\gamma_{M'} > 0$.
        \State Find any (maximal) alternating path (or cycle) $P$ in $M \mathbin{\triangle} M'$
        \For{each $e' \in M' \cap P$}
            \State Choose a facility $i_{e'}$ with $\tilde{x}_{i_{e'}, e'} > 0 $
        \EndFor
        \State $\varepsilon \gets \min\{\gamma_{M'}, \min_{e' \in M' \cap P} \{\tilde{x}_{i_{e'}, e'}\}\}$
        \For{each $e' \in M' \cap P$}
            \State Let $e \in M \cap P$ be the next edge to $e'$ in the order of $P$
            \State $\xtilde_{i_{e'}, e'} \gets \xtilde_{i_{e'}, e'} - \varepsilon$
            \State $\xtilde_{i_{e'}, e} \gets \xtilde_{i_{e'}, e} + \varepsilon$
        \EndFor
        \State $\gamma_{M'} \gets \gamma_{M'} - \varepsilon$
        \State $\gamma_{M' \triangle P} \gets \gamma_{M' \triangle P} + \varepsilon$
    \EndWhile
    \State \Return $(\tilde{x}, 2y)$
\end{algorithmic}  
\end{algorithm}

\def\cligap{1.5}
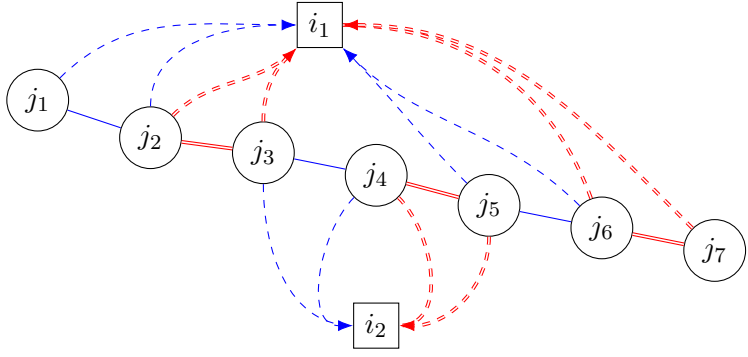
\begin{figure}[t]
    \centering
    \begin{tikzpicture}[
        fac/.style={
            draw=black,
            minimum width= 0.6cm,
            minimum height= 0.6cm
        },
        cli/.style={
            circle,
            draw=black,
            minimum width= 0.6cm,
            minimum height= 0.6cm
        },
        assigntip/.style={
            -{Latex[length=2mm]}
        },
        assign/.style={
            dashed
        },
        Morgn/.style={
            double,
            draw=red
        },
        Mprime/.style={
            draw=blue
        }
    ]        
        \node[fac] (i1) at (2.5*\cligap, 2) {$i_1$};
        \node[fac] (i2) at (3*\cligap, -2) {$i_2$};

        \node[cli] (j1) at (0*\cligap, 1.0) {$j_1$};
        \node[cli] (j2) at (1*\cligap, 0.5) {$j_2$};
        \node[cli] (j3) at (2*\cligap, 0.3) {$j_3$};
        \node[cli] (j4) at (3*\cligap, 0.0) {$j_4$};
        \node[cli] (j5) at (4*\cligap, -0.4) {$j_5$};
        \node[cli] (j6) at (5*\cligap, -0.7) {$j_6$};
        \node[cli] (j7) at (6*\cligap, -1.0) {$j_7$};

        \draw[Mprime] (j1) -- (j2);
        \draw[Mprime] (j3) -- (j4);
        \draw[Mprime] (j5) -- (j6);

        \draw[-] (j1) edge[assign, Mprime, out=45, in=180] (i1);
        \draw[assigntip] (j2) edge[assign, Mprime, out=90, in=180] (i1);
        \draw[-] (j3) edge[assign, Mprime, out=-90, in=180] (i2);
        \draw[assigntip] (j4) edge[assign, Mprime, out=-135, in=180] (i2);
        \draw[-] (j5) edge[assign, Mprime, out=135, in=-45] (i1);
        \draw[assigntip] (j6) edge[assign, Mprime, out=135, in=-45] (i1);

        \draw[Morgn] (j2) -- (j3);
        \draw[Morgn] (j4) -- (j5);
        \draw[Morgn] (j6) -- (j7);

        \draw[-] (j2) edge[assign, Morgn, out=45, in=225] (i1);
        \draw[assigntip] (j3) edge[assign, Morgn, out=90, in=225] (i1);
        \draw[-] (j4) edge[assign, Morgn, out=-45, in=0] (i2);
        \draw[assigntip] (j5) edge[assign, Morgn, out=-90, in=0] (i2);
        \draw[-] (j6) edge[assign, Morgn, out=110, in=0] (i1);
        \draw[assigntip] (j7) edge[assign, Morgn, out=135, in=0] (i1);
    \end{tikzpicture}
    \caption{Illustration of rerouting along an alternating path $P$. Facilities and clients are depicted by squares and circles, respectively. The alternating path is depicted by solid edges, where the blue single-line edges correspond to $P \cap M'$ and the red double-line edges correspond to $P \cap M$. The blue dotted single-line arrow represents a positive fractional assignment of each edge in $P \cap M'$ while the red dotted double-line arrow represents the rerouted assignment for each edge in $P \cap M$.}
    \label{fig:maxmat:reroute}

\end{figure}

\paragraph{Analysis.}
We aim at proving the following lemma.
\begin{lemma} \label{lem:maxmat:reroute}
    Given a feasible solution $(x, y)$ to \eqref{lp:flm} and a maximum matching $M \subseteq E$ in the compatibility graph $G = (V, E)$, \Reroute returns in polynomial time a solution $(\xtilde, \ytilde)$ satisfying
    \begin{itemize}
        \item $(\xtilde, \ytilde)$ is feasible to \eqref{lp:flm} where $\xtilde$ is supported by only $M$;
        \item $\open(\ytilde) \leq 2 \cdot \open(y)$ and $\connflm(\xtilde) \leq 2 \cdot \connflm(x) + \sum_{e \in M} d(e)$. 
    \end{itemize}
\end{lemma}

We first show the feasibility of the output.
\begin{lemma} \label{lem:maxmat:reroute:feas}
    In any iteration of \Reroute, $(\xtilde, 2y)$ is feasible to \eqref{lp:flm}.
\end{lemma}
\begin{proof}
    The fact that \eqref{lp:flm:pmm} and \eqref{lp:flm:nonneg} are satisfied follows from the choice of $\varepsilon = \min\{ \gamma_{M'}, \min_{e' \in M' \cap P} \{\xtilde_{i_{e'}, e'} \}\}$.
    To see \eqref{lp:flm:flow} is also satisfied, we fix a facility $i \in F$ and a client $j \in V$. 
    If $j$ is unmatched by $M$, note that every $\xtilde_{i, e}$ for every $e \in \delta(j)$ only decreases by the execution of the subroutine, implying that \eqref{lp:flm:flow} remain satisfied.
    Let us thus assume from now that $j$ is matched by $M$ with, say, $k \in V$.
    In this case, $\xtilde_{i, \{j, k\}}$ can increase, but this increment is always bounded from above by $\sum_{e' \in \delta(k) : e' \neq \{j, k\}} x_{i, e'}$ because the subroutine takes an alternating path (or cycle) at every iteration.
    We can thus deduce that
    \[
        \textstyle \sum_{e \in \delta(j)} \xtilde_{i, e}
        \leq \sum_{e \in \delta(j)} x_{i, e} + \sum_{e' \in \delta(k)} x_{i, e'}
        \leq 2 y_i,
    \]
    where the second inequality is due to the feasibility of the initial input $(x, y)$.
\end{proof}

Next, we argue that \Reroute returns a solution supported by $M$ in polynomial time.
\begin{lemma} \label{lem:maxmat:reroute:suppM}
    At termination, $\xtilde$ is supported by $M$.
    Moreover, \Reroute runs in polynomial time.
\end{lemma}
\begin{proof}
    To prove this lemma, we define a \emph{potential} of the current assignment $\xtilde \in \R^{F \times E}$ that only decreases over iterations, while $\xtilde$ is supported only by $M$ if its potential value is 0.
    Recall from the execution of \Reroute that the form of the next $\xtilde$ depends on the value of $\varepsilon$.
    If $\varepsilon = \gamma_{M'}$, we can view the update of $\gamma$ as replacing $M'$ entirely with $M' \mathbin{\triangle} P$ containing strictly fewer edges not in $M$.
    On the other hand, if $\varepsilon = \min_{e' \in M' \cap P} \{\xtilde_{i_{e'}, e'}\}$, the subroutine may not entirely replace $M'$ with $M' \mathbin{\triangle} P$ in $\gamma$.
    However, in this case, we can find that at least one edge not in $M$ becomes dropped from the support of $\xtilde$, while never getting back again.
    
    We exploit these observations to formulate the potential. Formally, we define the potential of $\xtilde \in \R^{F \times E}_+$ and its coefficient vector $\gamma \in \R^\calM_+$, denoted by $\Phi(\xtilde, \gamma)$, as follows:
    \begin{align*}
        \Phi(\xtilde, \gamma) 
        &
        := (|E|+1) \cdot |\{ (i, e') \in F \times (E \setminus M) : \xtilde_{i, e'} > 0 \}| 
        \textstyle + \sum_{M' \in \calM : \gamma_{M'} > 0} |M' \setminus M|.
    \end{align*}
    In other words, the potential incorporates the following two quantities:
    \begin{itemize}
        \item the number of tuples $(i, e') \in F \times (E \setminus M)$ such that compatible pair $e$ outside the input maximum matching $M$ is serviced by facility $i$ with a positive amount, multiplied by $(|E| + 1)$;
        \item the total number of edges in $M' \setminus M$ across all $M' \in \calM$ supporting $\gamma$.
    \end{itemize}
    Note that the potential is a nonnegative integer, and that the potential of the initial input is bounded from above by $O(|F|\cdot|E|^2)$ due to Proposition~\ref{prop:prelim:concom}.
    Furthermore, the number of nonzero entries in $\gamma$ is no greater than (the second term of) $\Phi(\xtilde, \gamma)$.
    Lastly, if $\Phi(\xtilde, \gamma) = 0$, this implies
    $
        \xtilde_{i, e'} = 0
    $
    for every $i \in F$ and $e' \in E \setminus M$, 
    meaning that $\xtilde$ is supported by $M$.

    It thus suffices to show that $\Phi(\xtilde, \gamma)$ only strictly decreases over iterations.
    Observe that, for every $i \in F$ and $e' \in E \setminus M$, $\xtilde_{i, e'}$ can only decrease in every iteration, and hence, the first term of the potential never increases.
    
    Fix an iteration.
    If $\varepsilon = \gamma_{M'}$, then $\gamma_{M'}$ becomes 0 while it may introduce $M' \mathbin{\triangle} P$ to the support of $\gamma$.
    However, since $P$ is an alternating path (or cycle) in $M \mathbin{\triangle} M'$, we have $|M' \setminus M| > |(M' \mathbin{\triangle} P) \setminus M|$, yielding that $\Phi(\xtilde, \gamma)$ decreases.
    On the other hand, if $\varepsilon = \min_{e' \in M' \cap P} \{\xtilde_{i_{e'}, e'}\}$, then the second term of the potential can increase by $|(M' \mathbin{\triangle} P) \setminus M| \leq |E|$.
    However, in this case, at least one tuple should be removed from $\{ (i, e') \in F \times (E \setminus M) : \xtilde_{i, e'} > 0 \}$ by the choice of $\varepsilon$, the first term of the potential would decrease by at least $|E| + 1$.
    Hence, $\Phi(\xtilde, \gamma)$ also strictly decreases in this case.
\end{proof}

We now show that the assignment cost of the output is bounded by twice the assignment cost of the given solution $x$ plus the cost of the given maximum matching $M$.
To gain more intuition, let us consider the situation of Fig.~\ref{fig:maxmat:reroute} for example.
We can view this situation as reassigning $j_3$ to $i_1$ instead of $j_1$, $j_5$ to $i_2$ instead of $j_3$, and $j_7$ to $i_1$ instead of $j_5$, respectively, along with the alternating path $P$.
Now, observe that the assignment cost of $j_3$ to $i_1$ is bounded from above by $d(j_3, j_2) + d(j_2, j_1) + d(j_1, i_1)$ due to the triangle inequality, where the first term is from the matched edge $\{j_3, j_2\} \in M$, while the second term and the third term (together with the assignment of $j_2$ to $i_1$) can be respectively charged to the assignment of $\{j_1, j_2\}$ to $i_1$ in the given solution $x$.
Accumulating these charges throughout the execution derives the following lemma.
\begin{lemma} \label{lem:maxmat:reroute:costbound}
    At termination,
    $
        \textstyle \connflm(\xtilde) \leq 2 \cdot \connflm(x) + \sum_{e \in M} d(e).
    $
\end{lemma}
\begin{proof}
    For $e \in M$, let us denote by $N(e) := \{ e' \in E \setminus M \mid e' \cap e \neq \emptyset \}$ the set of adjacent edges unmatched by $M$; similarly, for $e' \in E \setminus M$, we denote by $N(e') := \{e \in M \mid e \cap e' \neq \emptyset\}$ the set of adjacent matched edges.
    For $i \in F$, $e' \in E \setminus M$, and $e \in N(e')$, we further define $\transfer$ as the amount serviced by $i$ transferred from $e'$ to $e$ by the subroutine.
    Since $\xtilde$ is supported by $M$ at termination due to Lemma~\ref{lem:maxmat:reroute:suppM}, we have
    \begin{align}
        x_{i, e'} & \textstyle = \sum_{e \in N(e')} \transfer, && \text{ for all $i \in F$ and $e' \in E \setminus M$;} \label{eq:maxmat:reroute:costbound:xieprime} \\
        \xtilde_{i, e} & \textstyle = x_{i, e} + \sum_{e' \in N(e)} \transfer, && \text{ for all $i \in F$ and $e \in M$}. \label{eq:maxmat:reroute:costbound:xtildeie}
    \end{align}
    Moreover, for any $i \in F$, $e' \in E \setminus M$, and $e \in N(e')$, we have
    \begin{equation} \label{ineq:maxmat:reroute:costbound:dbound}
        d(i, e) \leq 2d(i, e') + d(e)
    \end{equation}
    because, when we say $e' = \{j, k\}$ and $e = \{k, \ell\}$, we can derive
    \begin{align*}
        d(i, e)
        &
        = d(i, k) + d(i, \ell)
        \leq d(i, k) + d(i, j) + d(j, k) + d(k, \ell)
        \\&
        = d(i, e') + d(e') + d(e)
        \leq 2d(i,e') + d(e),
    \end{align*}
    where the first inequality is due to the triangle inequality on $d(i, \ell)$ and the second inequality is due to Lemma~\ref{lem:prelim:distpairbound}.

    We can now prove the lemma using these (in)equalities. 
    \begin{align*}
        \connflm(\xtilde)
        &
        = \sum_{i \in F} \sum_{e \in E} d(i, e) \, \xtilde_{i, e}
        \stackrel{\textrm{(I)}}{=} \sum_{i \in F} \sum_{e \in M} d(i, e) \, \xtilde_{i, e}
        = \sum_{i \in F} \sum_{e \in M} d(i, e) \, \Big[ x_{i, e} + \sum_{e' \in N(e)} \transfer \Big]
        \\&
        = \sum_{i \in F} \sum_{e \in M} d(i, e) \, x_{i, e} + \sum_{i \in F} \sum_{e' \in E \setminus M} \sum_{e \in N(e')} d(i, e) \, \transfer
        \\&
        \stackrel{\textrm{(II)}}{\leq} \sum_{i \in F} \sum_{e \in M} d(i, e) \, x_{i, e}
        + \sum_{i \in F} \sum_{e' \in E \setminus M} 2d(i, e') \sum_{e \in N(e')} \transfer
        + \sum_{e \in M} d(e) \sum_{i \in F}\sum_{e' \in N(e)} \transfer
        \\&
        \stackrel{\textrm{(III)}}{\leq} \sum_{i \in F} \sum_{e \in E} 2d(i, e) \, x_{i, e}
        + \sum_{e \in M} d(e) \sum_{i \in F}  \sum_{e' \in N(e)} \transfer
        \\&
        \stackrel{\textrm{(IV)}}{\leq} \sum_{i \in F} \sum_{e \in E} 2d(i, e) \, x_{i, e}
        + \sum_{e \in M} d(e) \, \xtilde_e
        \stackrel{\textrm{(V)}}{\leq} 2 \cdot \connflm(x) + \sum_{e \in M} d(e),
    \end{align*}
    where (I) follows from the fact that $\xtilde$ is fully supported by $M$, (II) from \eqref{ineq:maxmat:reroute:costbound:dbound}, (III) from \eqref{eq:maxmat:reroute:costbound:xieprime}, (IV) from \eqref{eq:maxmat:reroute:costbound:xtildeie}, and (V) from the fact that $\xtilde_e \leq 1$ for every $e \in E$ due to \eqref{const:pmm:degree} of \eqref{lp:flm:pmm}.
\end{proof}

Lemma~\ref{lem:maxmat:reroute} then immediately follows from Lemmas~\ref{lem:maxmat:reroute:feas}, \ref{lem:maxmat:reroute:suppM}, and \ref{lem:maxmat:reroute:costbound}.

\paragraph{\Reroute for perfectly matchable $G$.}
If the compatibility graph $G$ is perfectly matchable, we can improve the performance of \Reroute by slightly modifying the execution of the subroutine.
We formally state the lemma for this case while deferring its proof to Appendix~\ref{app:perfmat:reroute}.
\begin{lemma} \label{lem:perfmat:reroute}
    Assume the compatibility graph $G = (V, E)$ is perfectly matchable. Given a feasible solution $(x, y)$ to \eqref{lp:flm} and a perfect matching $M$ in $G$, we can modify \Reroute to return in polynomial time a solution $(\xtilde, \ytilde)$ satisfying
    \begin{itemize}
        \item $(\xtilde, \ytilde)$ is feasible to \eqref{lp:flm} where $\xtilde$ is supported by only $M$;
        \item $\open(\ytilde) \leq \open(y)$ and $\connflm(\xtilde) \leq \connflm(x) + \sum_{e \in M} d(e)$. 
    \end{itemize}
\end{lemma}
\subsection{Main algorithm} \label{sec:maxmat:mainalg}

\paragraph{Algorithm description.}
We are now ready to describe our main algorithm that guarantees an approximation factor $3.868$ for \flm.
See Algorithm~\ref{alg:flm} for a pseudocode of this algorithm.

Given an \FLM instance $(F, V, E, f, d)$, we start with solving \eqref{lp:flm} to obtain an optimal solution $(\xstar, \ystar) \in \R^{F \times E}_+ \times \R^F_+$.
We also find a minimum-cost maximum matching $\Mstar \subseteq E$ in the compatibility graph $G$ where we regard the cost of each edge $e \in E$ as $d(e)$.
We then run $\Reroute((\xstar, \ystar), \Mstar)$; let $(\xtilde, \ytilde)$ denote the output of the subroutine.
As $\xtilde$ is supported only by $\Mstar$ due to Lemma~\ref{lem:maxmat:reroute}, we can regard $\xtilde$ as a vector in $\R^{F \times \Mstar}_+$.

We now construct an auxiliary \UFL instance $(F, \Mstar, f, d')$ where each compatible pair in $\Mstar$ is regarded as a ``meta-client'' and $d' : F \times \Mstar \to \R_+$ is defined as
$
    \textstyle d'(i, e) := d(i, e)
$
for every $i \in F$ and $e \in \Mstar$.
Then, for some $\lambda > 1.678$ to be chosen later, we run a $(\lambda, 1+\frac{2}{e^\lambda})$-approximation LP-rounding algorithm for \UFL, guaranteed by Proposition~\ref{prop:prelim:bifactor}, with $(\xtilde, \ytilde)$ input upon the auxiliary instance.
\footnote{One may want to first compute an optimal fractional solution to \eqref{lp:ufl} on the auxiliary \UFL instance $(F, \Mstar, f, d')$, and then run the LP-rounding algorithm with this solution. 
However, note that \eqref{lp:ufl} on $(F, \Mstar, f, d')$ may not be a feasible LP relaxation for the original \FLM instance, and hence, it is anyway necessary to address the approximation of the fractional \UFL solution to the optimal \FLM solution.
We instead compute the optimal solution to \eqref{lp:flm}, serving as a valid lower bound, and reroute it into a feasible solution to \eqref{lp:ufl} supported by $\Mstar$.}
For $S \subseteq F$ and $\sigma : \Mstar \to S$ output by the LP-rounding algorithm, we return $(S, \Mstar, \sigma)$ as a final solution.
This ends the description of our main algorithm.

\begin{algorithm}[t]
\caption{A $3.868$/$2.373$-approximation algorithm for \FLM}
\label{alg:flm}
\begin{algorithmic}[1]
    \State Solve \eqref{lp:flm} and obtain an optimal solution $(\xstar, \ystar)$ \label{alg:step1}
    \State Compute a minimum-cost maximum matching $\Mstar$ in $G$ with cost being $d$
    \State $(\xtilde, \ytilde) \gets \Reroute((\xstar, \ystar), \Mstar)$
    \State Construct an auxiliary UFL instance $(F, \Mstar, d')$ with $d' = d$ \label{alg:ufl}

    \State Run a $(\lambda, 1+\frac{2}{e^\lambda})$-approximation LP-rounding algorithm for UFL with $(\xtilde, \ytilde)$ on $(F, \Mstar, d')$ to obtain $(S, \sigma)$
    \State \Return $(S, \Mstar, \sigma)$
\end{algorithmic}  
\end{algorithm}

\paragraph{Analysis.}
We state the main theorems of this section:
\begin{theorem} \label{thm:maxmat:main}
    Algorithm~\ref{alg:flm} is a $\max\{2 \lambda, 3 \, (1 + \frac{2}{e^\lambda})\}$-approximation algorithm for \FLM for $\lambda \geq 1.678$.
    The integrality gap of \eqref{lp:flm} is also bounded from above by $\max\{2 \lambda, 3 \, (1 + \frac{2}{e^\lambda})\}$ for $\lambda \geq 1.678$.
\end{theorem}
\begin{theorem} \label{thm:maxmat:main:perfmat}
    If $G$ is perfectly matchable, Algorithm~\ref{alg:flm} is a $\max\{\lambda, 2 \, (1 + \frac{2}{e^\lambda})\}$-approximation algorithm for \FLM for $\lambda \geq 1.678$.
    The integrality gap of \eqref{lp:flm} is also bounded from above by $\max\{\lambda, 2 \, (1 + \frac{2}{e^\lambda})\}$ for $\lambda \geq 1.678$ in this case.
\end{theorem}
Therefore, by choosing $\lambda = 1.934$, we can in general guarantee a $3.868$-approximation algorithm as well as the same bound on the integrality gap of \eqref{lp:flm}.
On the other hand, if $G$ is perfectly matchable, by choosing $\lambda=2.373$, the approximation ratio and the bound on the integrality gap are improved to $2.373$.

The following lemmas will be useful in proving the main theorem.
\begin{lemma} \label{lem:maxmat:Mstarbound}
    $\sum_{e \in \Mstar} d(e) \leq \connflm(\xstar)$.
\end{lemma}
\begin{proof}
    By Lemma~\ref{lem:prelim:distpairbound}, we have
    \[
        \connflm(\xstar)
        = \sum_{i \in F} \sum_{e \in E} d(i, e) \, x_{i, e}
        \geq \sum_{e \in E} d(e) \, x_e.
    \]
    Since $\{x_e\}_{e \in E} \in \PMM(G)$ due to \eqref{lp:flm:pmm} and $\Mstar$ is a minimum-cost maximum matching with cost being $d$, the right-hand side is lower-bounded by the cost of $\Mstar$, completing the proof.
 \end{proof}

\begin{lemma} \label{lem:maxmat:threehop}
    The assignment cost function $d' : F \times \Mstar \to \R_+$ of the auxiliary \UFL instance satisfies the three-hop inequality.
\end{lemma}
\begin{proof}
    For any $i, i' \in F$ and $e, e' \in \Mstar$, we say $e = \{j, k\}$ and $e' = \{j', k'\}$.
    We then have
    \begin{align*}
        d'(i, e)
        &
        = d(i, j) + d(i, k)
        \\&
        \leq (d(i, j') + d(j', i') + d(i', j)) + (d(i, k') + d(k', i') + d(i', k))
        \\&
        = (d(i, j') + d(i, k')) + (d(i', j') + d(i', k')) + (d(i, j') + d(i, k'))
        \\&
        = d'(i, e') + d'(i', e') + d'(i', e),
    \end{align*}
    where the inequality follows from the triangle inequality of $d$.
\end{proof}

\begin{lemma} \label{lem:maxmat:feasufl}
    The output $(\xtilde, \ytilde)$ of $\Reroute((\xstar, \ystar), \Mstar)$ is a feasible solution to \eqref{lp:ufl}.
\end{lemma}
\begin{proof}
    Since $\xtilde$ is supported by only $\Mstar$ due to Lemma~\ref{lem:maxmat:reroute}, we can assume $\xtilde \in \R^{F \times \Mstar}_+$.
    We can also guarantee $\xtilde_e = 1$ for every $e \in \Mstar$ due to \eqref{lp:flm:pmm} of \eqref{lp:flm}, implying that $(\xtilde, \ytilde)$ satisfies \eqref{lp:ufl:assign} of \eqref{lp:ufl}.
    To see \eqref{lp:ufl:open} of \eqref{lp:ufl} is also satisfied, let us fix a matched edge $e \in \Mstar$.
    We can observe that, for any incident client $j \in e$,
    \[
        \textstyle \xtilde_{i, e} = \sum_{e' \in \delta(j)} \xtilde_{i, e'} \leq y_i,
    \]
    where the equality is because $\xtilde$ is supported by only $\Mstar$, and the inequality is due to \eqref{lp:flm:flow} of \eqref{lp:flm}.
    Finally, \eqref{lp:ufl:nonneg} is trivially satisfied.
\end{proof}

We are now ready to prove the main theorems.

\begin{proof}[Proof of Theorem~\ref{thm:maxmat:main}]
    Lemmas~\ref{lem:maxmat:threehop} and \ref{lem:maxmat:feasufl} imply that the algorithm is well-defined.
    Let us now turn to analyzing the approximation guarantee of the algorithm.
    By Lemmas~\ref{lem:maxmat:reroute} and \ref{lem:maxmat:Mstarbound}, we have
    \begin{equation} \label{ineq:maxmat:main:reroutebound}
        \textstyle \open(\ytilde) \leq 2 \cdot \open(\ystar)
        \text{ and }
        \connflm(\xtilde) \leq 3 \cdot \connflm(\xstar).
    \end{equation}
    Moreover, as $\xtilde$ is supported by $\Mstar$, the assignment cost incurred by $(\xtilde, \ytilde)$ in the auxiliary \UFL instance is identical to that in the original \FLM instance, i.e.,
    $
        \textstyle
        \connufl(\xtilde)
        = \sum_{i \in F} \sum_{e \in \Mstar} d'(i, e) \, \xtilde_{i, e} 
        = \sum_{i \in F} \sum_{e \in E} d(i, e) \, \xtilde_{i, e}
        = \connflm(\xtilde).
    $
    We can therefore bound the expected total cost of our solution as follows:
    \begin{align}
        &
        \textstyle  \E \Big[ \sum_{i \in S} f_i + \sum_{e \in \Mstar} d(\sigma(e), e) \Big]
        = \E \Big[\sum_{i \in S} f_i + \sum_{e \in \Mstar} d'(\sigma(e), e) \Big]
        \nonumber\\&
        \textstyle \quad \leq \lambda \cdot \open(\ytilde) + \Big(1 + \frac{2}{e^\lambda} \Big) \cdot \connufl(\xtilde) 
        \nonumber\\& 
        \textstyle \quad \leq \max \Big\{ 2 \lambda, 3 \Big(1 + \frac{2}{e^\lambda} \Big) \Big\} \cdot (\open(\ystar) + \connflm(\xstar) ), \label{ineq:maxmat:main:costbound}
    \end{align}
    where the first inequality is due to the bifactor-approximation LP-rounding algorithm.
    This also implies that the integrality gap of \eqref{lp:flm} is bounded by the same factor.
\end{proof}

We can prove Theorem~\ref{thm:maxmat:main:perfmat} using almost the same argument with small modification where \eqref{ineq:maxmat:main:reroutebound} is replaced with 
$\open(\ytilde) \leq \open(\ystar)$ and $\connflm(\xtilde) \leq 2 \cdot \connflm(\xstar)$, and hence, the right-hand side of \eqref{ineq:maxmat:main:costbound} can be substituted by 
$
    \textstyle \max \{ \lambda, 2 (1 + \frac{2}{e^\lambda} ) \} \cdot (\open(\ystar) + \connflm(\xstar) ).
$
\section{2.218-approximation algorithm for \FLM with perfectly matchable compatibility graph} \label{sec:perfmat}

In the previous section, we present a $2.373$-approximation algorithm for \flm where the compatibility graph is perfectly matchable.
In this section, we present another algorithm guaranteeing an improved approximation factor of $2.218$ for this case.

\paragraph{Algorithm description.}
To describe this algorithm, we need a few more notation.
For a subset $S \subseteq F$ of facilities, let
$d(S, j) := \min_{i \in S} \{ d(i, j) \}$ for $j \in V$, and $d(S, e) := \min_{i \in S} \{ d(i, e) \}$ for $e \in E$.
Next, for any $x \in \R^{F \times E}$, 
$x_{i,j} := \sum_{e \in \delta(j)} x_{i, e}$ for $i \in F$ and $j \in V$; 
recall that $x_e := \sum_{i \in F} x_{i, e}$ for $e \in E$.

We are now ready to describe the algorithm; the pseudocode of this algorithm is provided in Algorithm~\ref{alg:perfectflm}. 
Given an \FLM instance $(F, V, E, f, d)$, we start with solving \eqref{lp:flm} to obtain an optimal solution $(\xstar, \ystar) \in \R^{F \times E}_+ \times \R^F_+$.
We then construct $( \{\xstar_{i ,j}\}_{i \in F, j \in V}, \ystar) \in \R^{F \times V}_+ \times \R^F_+$, and run a $(\lambda, 1 + \frac{2}{e^\lambda})$-approximation LP-rounding algorithm for \UFL with this constructed solution on an auxiliary \UFL instance $(F, V, f, d)$ for some $\lambda \geq 1.678$ to be chosen.
Let $S \subseteq F$ be the set of open facilities.
We then compute a minimum-cost perfect matching $\Mstar_S \subseteq E$ in the compatibility graph $G=(V,E)$ where the cost of each edge $e \in E$ is defined as $d(S, e)$.
Let $\sigma_S : \Mstar_S \to S$ be defined as $\sigma_S (e) = \argmin_{i \in S} d(i, e)$ for every $e \in \Mstar_S$.
We then return $(S, \Mstar_S, \sigma_S)$.

\begin{algorithm}[t]
\caption{A 2.218-approximation algorithm for \FLM with a perfectly matchable compatibility graph}
 \label{alg:perfectflm}
 \begin{algorithmic}[1]

    \State Solve \eqref{lp:flm} and obtain an optimal solution  $(\xstar, \ystar)$
    \State Run the $(\lambda, 1 + \frac{2}{e^\lambda})$-approximation LP-rounding algorithm for \UFL with $(\{\xstar_{ij}\}_{i \in F, j \in V}, \ystar)$ on $(F, V, f, d)$ to obtain $S \subseteq F$
    \State Compute a minimum-cost perfect matching $\Mstar_S$ in $G$ with cost being $d(S, \cdot)$
    \State Let $\sigma_S(e) := \argmin_{i \in S} d(i, e)$ for $e \in \Mstar_S$
    \State \Return $(S, \Mstar_S, \sigma_S)$
\end{algorithmic}
\end{algorithm}

\paragraph{Analysis.}
The following lemmas will be useful in proving the main theorem.
\begin{lemma} \label{lem:perfmat:uflapprox}
    The constructed solution $(\{\xstar_{i ,j}\}_{i \in F, j \in D}, \ystar)$ is feasible to \eqref{lp:ufl} with respect to $(F, V, f, d)$, and its total assignment cost is the same as the total assignment cost incurred by $(\xstar, \ystar)$ in the given \FLM instance $(F, V, E, f, d)$, i.e., $\connufl(\{\xstar_{i, j}\}_{i \in F, j \in V}) = \connflm(\xstar)$.
\end{lemma} 
\begin{proof}
    We first show the feasibility of the constructed solution. For \eqref{lp:ufl:assign}, observe that
    \[
        \textstyle 
        \sum_{i \in F} \xstar_{i, j} 
        = \sum_{i \in F} \sum_{e \in \delta(j)} \xstar_{i, j}
        = 1,
    \]
    where the second equality follows from \eqref{const:ppm:degree} of \eqref{lp:flm:pmm}.
    Note also that \eqref{lp:ufl:open} and \eqref{lp:ufl:nonneg} are satisfied due to \eqref{lp:flm:flow} and \eqref{lp:flm:nonneg}, respectively.
    
    The second half of the statement can be shown as follows:
    \begin{align*}
        \connufl(\{\xstar_{i, j}\}_{i \in F, j \in V})
        &
        \textstyle = \sum_{i \in F} \sum_{j \in V} d(i, j) \xstar_{i, j}
        \textstyle = \sum_{i \in F} \sum_{j \in V} \sum_{e \in \delta(j)} d(i, j) \, \xstar_{i, e}
        \\&
        \textstyle = \sum_{i \in F} \sum_{e = \{j, k\} \in E} (d(i, j) + d(i, k)) \xstar_{i, e}
        = \connflm(\xstar)
    \end{align*}
\end{proof}

\begin{lemma} \label{lem:perfmat:dSebound}
    For $S \subseteq F$ and $e = \{j, k\} \in E$, $d(S, e) \leq d(e) + d(S, j) + d(S, k)$.
\end{lemma}
\begin{proof}
    Let $i_j := \argmin_{i \in S} d(i, j)$ and $i_k := \argmin_{i \in S} d(i, k)$.
    Since $i_j, i_k \in S$, we have
    \begin{align*}
        d(S, e) & \leq d(i_j, j) + d(i_j, k) \leq 2d(i_j, j) + d(j, k); \\
        d(S, e) & \leq d(i_k, j) + d(i_k, k) \leq 2d(i_k, k) + d(j, k),
    \end{align*}
    yielding that
    \[
        d(S, e) \leq d(j, k) + d(i_j, j) + d(i_k, k) = d(e) + d(S, j) + d(S, k).
    \]
\end{proof}

\begin{lemma} \label{lem:perfmat:solconbound}
    The assignment cost of the output is bounded by 
    $\sum_{e \in E} d(S, e) \, \xstar_e.$
\end{lemma}
\begin{proof}
    Observe that
    \[
        \sum_{e \in \Mstar_S} d(\sigma_S(e), e)
        = \sum_{e \in \Mstar_S} d(S, e)
        \leq \sum_{e \in E} d(S, e) \, \xstar_e,
    \]
    where the equality is due to the definition of $\sigma_S$, and the inequality comes from that $\Mstar_S$ is a minimum-cost perfect matching with cost being $d(S, \cdot)$ and $\{\xstar_e\}_{e \in E} \in \PMM(G)$ due to \eqref{lp:flm:pmm} of \eqref{lp:flm}.
\end{proof}

We are now ready to prove the main theorem of this section:
\begin{theorem}\label{thm:perfmat:main}
    Algorithm~\ref{alg:perfectflm} is a $\max\{\lambda, 2 + \frac{2}{e^\lambda}\}$-approximation algorithm for $\lambda \geq 1.678$.
    The integrality gap of \eqref{lp:flm} is at most $\max\{\lambda, 2 + \frac{2}{e^\lambda}\}$ when the compatibility graph is perfectly matchable for $\lambda \geq 1.678$.
\end{theorem}
When we choose $\lambda = 2.218$, this theorem implies that both the approximation factor and the integrality gap of \eqref{lp:flm} are (at most) $2.218$ in the case where a perfectly matchable compatibility graph is given.
\begin{proof}[Proof of Theorem~\ref{thm:perfmat:main}]
    The expected total cost incurred by the algorithm's solution is bounded by
    \begin{align}
        & 
        \textstyle \E \Big[ \sum_{i \in S} f(i) + \sum_{e \in \Mstar_S} d(\sigma_S(e), e) \Big]
        \leq \E \Big[ \sum_{i \in S} f(i) + \sum_{e \in E} d(S, e) \xstar_e \Big]
        \nonumber \\&
        \textstyle \quad \leq \E \Big[ \sum_{i \in S} f(i)
        + \underbrace{\textstyle \sum_{e \in E} d(e) \xstar_e}_{\textrm{(I)}}
        + \underbrace{\textstyle \sum_{\{j, k\} \in E} (d(S, j) + d(S, k)) \xstar_{\{j, k\}}}_{\textrm{(II)}} \Big], \label{ineq:perfmat:main:01}
    \end{align}
    where the first inequality follows from Lemma~\ref{lem:perfmat:solconbound} and the second from Lemma~\ref{lem:perfmat:dSebound}.
    We can further bound (I) as follows:
    \[
        \textstyle 
        \textrm{(I)}
        = \sum_{i \in F} \sum_{e \in E} d(e)  \xstar_{i, e}
        \leq \sum_{i \in F} \sum_{e \in E} d(i, e)  \xstar_{i, e}
        = \connflm(\xstar),
    \]
    where the inequality is due to Lemma~\ref{lem:prelim:distpairbound}.
    In the meanwhile, (II) can be rephrased as follows:
    \[
        \textstyle
        \textrm{(II)}
        = \sum_{j \in V} d(S, j) \Big[ \sum_{e \in \delta(j)} \xstar_e \Big]
        = \sum_{j \in V} d(S, j),
    \]
    where the second equality is due to \eqref{const:ppm:degree} of \eqref{lp:flm:pmm}.

    We can thus further expand \eqref{ineq:perfmat:main:01} as follows:
    \begin{align*}
        \textstyle \E \Big[ \sum_{i \in S} f(i) + \sum_{e \in \Mstar_S} d(\sigma_S(e), e) \Big]
        &
        \textstyle \quad \leq \E \Big[ \sum_{i \in S} f(i) + \sum_{j \in V} d(S, j) \Big] + \connflm(\xstar)
        \\&
        \quad \leq \lambda \cdot \open(\ystar) + \big( {\textstyle 1 + \frac{2}{e^\lambda}} \big) \cdot \connufl(\{\xstar_{i, j}\}_{i \in F, j \in V}) + \connflm(\xstar)
        \\&
        \quad = \lambda \cdot \open(\ystar) + \big( {\textstyle 2 + \frac{2}{e^\lambda} } \big) \cdot \connflm(\xstar).
        \\&
        \quad \leq \max \big\{\lambda,{\textstyle 2 + \frac{2}{e^\lambda} } \big\} \cdot \big( \open(\ystar) + \connflm(\xstar) \big),
    \end{align*}
    where the second inequality and the equality follow from Lemma~\ref{lem:perfmat:uflapprox}.
    This also proves the same bound on the integrality gap of \eqref{lp:flm} when the compatibility graph is perfectly matchable.
\end{proof}
\section{Discussion and future directions} \label{sec:concl}
An obvious direction is to improve the approximation ratio for \flm.
In particular, compared to the factor attained for the perfectly matchable case, the approximation ratio of our algorithm for the general case is substantially higher.
This gap is mainly due to the multiplicative factors on the opening and connection costs in Lemma~\ref{lem:maxmat:reroute}.

Unfortunately, there are some fundamental limitations in our approach.
We first argue that the factor on the opening cost should be at least $\frac{3}{2}$.
Given a triangle as the compatibility graph, consider a fractional solution where each compatible pair is serviced $\frac{1}{3}$ by a facility opened $\frac{2}{3}$.
It is easy to see that this solution is feasible to \eqref{lp:flm}.
However, when we reroute it into one supported by any maximum matching, the facility must be integrally open, showing that the factor on the opening cost should be at least $\frac{3}{2}$.
We believe that this $\frac{3}{2}$ is indeed the correct factor for the opening cost.

On the other hand, the factor $3$ on the connection cost is also ``locally'' tight in the analysis.
For example, given three clients $v_1, v_2, v_3$ that are collinear and uniformly spaced in that order, if we reroute the assignment of $\{v_1, v_2\}$ serviced by a facility co-located at  $v_1$ into $\{v_2, v_3\}$, the connection cost increases by three times.
To overcome this analytical limitation, it would be interesting to find a way of bounding the increase of the connection cost in an amortized sense.

Many match-making systems would prefer high service quality for user experience in a modest tradeoff in the matching size.
This situation can be modeled by a robust or prize-collecting setting.
To this end, it is interesting to design a greedy primal-dual algorithm \cite{jain2001approximation,jain2002new} because such primal-dual algorithms have proven to be useful in related settings such as robust clustering~\cite{charikar2001algorithms,jain2003greedy} and  partial covering~\cite{gandhi2001approximation}.
In this direction, we believe that Hungarian methods would play a central role since \FLM generalizes \textsc{Minimum-cost Maximum Matching}.
A challenge is however in bypassing the monotone dual growth, which is not the case in Hungarian methods in general.

Lastly, it would also be very interesting to study \FLM under dynamic/online models or to extend the matching constraint to a more general $k$-packing constraint.

\section*{Acknowledgment}
The authors thank the anonymous reviewers for their valuable comments.
This work was supported by Polish National Science Centre grants 2020/39/B/ST6/01641, 2022/45/B/ST6/00559 and 2022/47/D/ST6/02864. This work was done while MM was a visiting intern at the University of Wroc\l{}aw, Poland.

\bibliographystyle{plain}
\bibliography{ref}

\appendix
\section{Deferred proof of Proposition~\ref{prop:prelim:pmm}} \label{app:deferpf}

To prove the first half of Proposition~\ref{prop:prelim:pmm}, we show the following:
\begin{lemma} 
    For any graph $G = (V, E)$, if $z \in \R^E$ satisfies \eqref{const:pmm:degree}-\eqref{const:pmm:nonneg}, $z$ can be represented by a convex combination of maximum matchings in $G$, implying that $z \in \PMM(G)$.
\end{lemma}
\begin{proof}
    We follow the standard technique: making a copy of $G$, denoted by $G' := (V', E')$, and creating an edge between each vertex $v \in V$ with its copy $v' \in V'$.
    We denote by $E'' := \{\{v, v'\}\}_{v \in V}$ these edges connecting $V$ and $V'$.
    Observe that $\Gtilde := (V \cup V', E \cup E' \cup E'')$ is perfectly matchable.
    
    Consider $\ztilde \in \R^{E \cup E' \cup E''}$ defined as follows:
    \[
        \ztilde_f := \begin{cases}
            z_e, & \text{if $f = e \in E$}; \\
            z_e, & \text{if $f = e' \in E'$}; \\
            1 - \sum_{e \in \delta(v)} z_e, & \text{if $f = \{v, v'\} \in E''$}.
        \end{cases}
    \]
    As $z$ satisfies \eqref{const:pmm:degree}, \eqref{const:pmm:oddset}, and \eqref{const:pmm:nonneg}, it is known that $\ztilde$ satisfies \eqref{const:ppm:degree}, \eqref{const:ppm:oddset}, and \eqref{const:ppm:nonneg} (see, e.g., Corollary 25.1a of Schrijver~\cite{schrijver2003combinatorial}).
    We can therefore see that $\ztilde$ can be represented by a convex combination of perfect matchings in $\Gtilde$.
    Let $\gammatilde \in \R^{\calM(\Gtilde)}_+$ be its coefficient vector.

    We further argue that, for any $\Mtilde \in \calM(\Gtilde)$ with $\gammatilde_{\Mtilde} > 0$, we must have $|\Mtilde \cap E''| = |V| - 2 \nu(G)$.
    Observe first that we have
    \[
        \textstyle 
        \sum_{f \in E''} \ztilde_f
        = \sum_{v \in V} \Big( 1 - \sum_{e \in \delta(v)} z_e \Big)
        = |V| - 2 \sum_{e \in E} z_e
        = |V| - 2 \nu(G),
    \]
    where the last equality follows from \eqref{const:pmm:size}.
    Therefore, if there exists a perfect matching $\Mtilde$ in $\Gtilde$ with $\gammatilde_{\Mtilde} > 0$ and $|\Mtilde \cap E''| > |V| - 2 \nu(G)$,
    there must exist another perfect matching $\Mtilde'$ in $\Gtilde$ with $\gammatilde_{\Mtilde'} > 0$ and $|\Mtilde' \cap E''| < |V| - 2 \nu(G)$.
    However, this is impossible since $\Mtilde' \cap E$ is a matching in $G$ of size $|\Mtilde' \cap E| = \frac{|V| - |\Mtilde' \cap E''|}{2} > \nu(G)$.
    We can therefore conclude that, for any $\Mtilde \in \calM(\Gtilde)$ with $\gammatilde_{\Mtilde} > 0$, $\Mtilde \cap E$ is a maximum matching in $G$, and hence, we can losslessly reduce $\gammatilde$ on $\R^{\calM(G)}_+$, completing the proof of the lemma.
\end{proof}
\section{\Reroute for perfectly matchable compatibility graph} \label{app:perfmat:reroute}
\paragraph{Subroutine description.}
Let us only mention the steps modified from the subroutine for the general case presented in Section~\ref{sec:maxmat:reroute}; see Algorithm~\ref{alg:perfmat:reroute} for a pseudocode of this modified subroutine.
Note that, as $G$ is now perfectly matchable, we can always find an alternating cycle $C$ in $M \mathbin{\triangle} M'$.
Then, for each $e' \in M' \cap C$, instead of transferring $\varepsilon$ of $\xtilde_{i_{e'}, e'}$ to only one of the adjacent edges (which is the case in the general version), we split $\varepsilon$ in half and transfer this $\frac{\varepsilon}{2}$ to both adjacent edges, respectively; more precisely, we update
\[
    \textstyle
    \xtilde_{i_{e'}, e'} \gets \xtilde_{i_{e'}, e'} - \varepsilon,
    \;
    \xtilde_{i_{e'}, e_1} \gets \xtilde_{i_{e'}, e_1} + \frac{\varepsilon}{2},
    \text{ and }
    \xtilde_{i_{e'}, e_2} \gets \xtilde_{i_{e'}, e_2} + \frac{\varepsilon}{2},
\]
where we denote by $e_1, e_2 \in M \cap C$ the adjacent edges with $e' \in M' \cap C$ in cycle $C$.
Once we finish rerouting the assignment $\xtilde$, the final output is $(\xtilde, y)$.
This is the end of the modification.

\begin{algorithm}[t]
\caption{$\Reroute((x, y), M)$ for a perfectly matchable compatibility graph} \label{alg:perfmat:reroute}
\begin{algorithmic}[1]
    \State $\xtilde \gets x$
    \State Let $\gamma \in \R^\calM_+$ be the coefficient vector of a decomposition of $\{ \xtilde_e \}_{e \in E}$ into a convex combination of perfect matchings
    \While{$\gamma_M < 1$}
        \State Let $M' \neq M$ be any perfect matching with $\gamma_{M'} > 0$.
        \State Find any alternating cycle $C$ in $M \mathbin{\triangle} M'$
        \For{each $e' \in M' \cap C$}
            \State Choose a facility $i_{e'}$ with $\tilde{x}_{i_{e'}, e'} > 0 $
        \EndFor
        \State $\varepsilon \gets \min\{\gamma_{M'}, \min_{e' \in M' \cap C} \{\tilde{x}_{i_{e'}, e'}\}\}$
        \For{each $e' \in M' \cap C$}
            \State Let $e_1, e_2 \in M \cap C$ be the adjacent edges with $e'$ in $C$
            \State $\xtilde_{i_{e'}, e'} \gets \xtilde_{i_{e'}, e'} - \varepsilon$
            \State $\xtilde_{i_{e'}, e_1} \gets \xtilde_{i_{e'}, e_1} + \frac{\varepsilon}{2}$
            \State $\xtilde_{i_{e'}, e_2} \gets \xtilde_{i_{e'}, e_2} + \frac{\varepsilon}{2}$
        \EndFor
        \State $\gamma_{M'} \gets \gamma_{M'} - \varepsilon$
        \State $\gamma_{M' \mathbin{\triangle} P} \gets \gamma_{M' \mathbin{\triangle} P} + \varepsilon$
    \EndWhile
    \State \Return $(\tilde{x}, y)$
\end{algorithmic}  
\end{algorithm}

\paragraph{Analysis.}
We restate the lemma we prove in this appendix.
\begin{lemma}[cf. Lemma~\ref{lem:perfmat:reroute}] \label{lem:app:perfmat:reroute}
    For a perfectly matchable compatibility graph $G = (V, E)$, this modified \Reroute returns in polynomial time a solution $(\xtilde, \ytilde)$ satisfying
    \begin{enumerate}
        \item $(\xtilde, \ytilde)$ is feasible to \eqref{lp:flm};
        \item $\xtilde$ is supported by only $M$;
        \item $\connflm(\xtilde) \leq \connflm(x) + \sum_{e \in M} d(e)$;
        \item $\open(\ytilde) \leq \open(y)$.
    \end{enumerate}
\end{lemma}

The next two lemmas can be shown using the same arguments in the proofs of Lemmas~\ref{lem:maxmat:reroute:feas} and \ref{lem:maxmat:reroute:suppM}, respectively.
\begin{lemma} \label{lem:app:perfmat:reroute:feas01}
    In any iteration, $\xtilde$ satisfies \eqref{lp:flm:pmm} and \eqref{lp:flm:nonneg} of \eqref{lp:flm}.
\end{lemma}

\begin{lemma} \label{lem:app:perfmat:reroute:suppM}
    At termination, $\xtilde$ is supported by $M$.
    Moreover, \Reroute runs in polynomial time.
\end{lemma}

The next lemma is crucial to get improvement over the general case.
\begin{lemma} \label{lem:app:perfmat:reroute:feas02}
    At termination, we have
    \begin{itemize}
        \item $(\xtilde, y)$ satisfies \eqref{lp:flm:flow} of \eqref{lp:flm};
        \item $\connflm(\xtilde) \leq \connflm(x) + \sum_{e \in M} d(e)$.
    \end{itemize}    
\end{lemma}
\begin{proof}
    By Lemma~\ref{lem:app:perfmat:reroute:suppM}, $\xtilde$ is supported by only $M$ at termination.
    Moreover, at every iteration, the subroutine transfers each half of the decrement by $\varepsilon$ on $\xtilde_{i_{e'}, e'}$ for each $e' \in M' \cap C$ towards the adjacent edges in $M \cap C$, respectively.
    We can therefore observe that, at termination, for any $i \in F$ and $e = \{j, k\} \in M$, exactly $\frac{1}{2} \sum_{e' \in \delta(j) \setminus e} x_{i, e'}$ and $\frac{1}{2} \sum_{e' \in \delta(k) \setminus e} x_{i, e'}$ are transferred into $\xtilde_{i, e}$ from the edges in $\delta(j) \setminus e$ and $\delta(k) \setminus e$, respectively.
    In other words, we have
    \begin{align}
        \xtilde_{i, e}
        &
        \textstyle = x_{i, e} + \frac{1}{2} \Big[ \sum_{e' \in \delta(j) \setminus e} x_{i, e'} + \sum_{e' \in \delta(k) \setminus e} x_{i, e'}  \Big] \label{eq:app:perfmat:reroute:id01}
        \\&
        \textstyle = \frac{1}{2} \Big[ \sum_{e' \in \delta(j)} x_{i, e'} + \sum_{e' \in \delta(k)} x_{i, e'}  \Big]. \label{eq:app:perfmat:reroute:id02}
    \end{align}

    The first statement can be shown as follows: for any $i \in F$ and $j \in V$ with $e = \{j, k\} \in M$,
    \[
        \textstyle \sum_{e'' \in \delta(j)} \xtilde_{i, e''}
        = \xtilde_{i, e}
        =  \frac{1}{2} \Big[ \sum_{e' \in \delta(j)} x_{i, e'} + \sum_{e' \in \delta(k)} x_{i, e'}  \Big]
        \leq y_i,
    \]
    where the first equality is due to the fact that $\xtilde$ is supported by $M$, the second eqaulity is due to \eqref{eq:app:perfmat:reroute:id02}, and the inequality follows from that the initial input $(x, y)$ satisfies \eqref{lp:flm:flow} of \eqref{lp:flm}.

    We now turn to proving the second statement. Recall that, for any $e' \in E \setminus M$, we use $N(e')$ to denote the adjacent matched edges by $M$.
    We can then observe:
    \begin{align*}
        &
        \connflm(\xtilde)
        = \sum_{i \in F} \sum_{e \in E} d(i, e) \, \xtilde_{i, e}
        \\&
        \quad\stackrel{\textrm{(I)}}{=} \sum_{i \in F} \sum_{e \in M} d(i, e) \, \xtilde_{i, e}
        \\&
        \quad \stackrel{\textrm{(II)}}{=} \sum_{i \in F} \sum_{e=\{j,k\} \in M} d(i, e) \, \Big[ x_{i, e} + {\textstyle \frac{1}{2}} \Big[ \sum_{e' \in \delta(j) \setminus e} x_{i, e'} + \sum_{e' \in \delta(k) \setminus e} x_{i, e'}  \Big] \Big]
        \\&
        \quad = \sum_{i \in F} \sum_{e \in M} d(i, e) \, x_{i, e} + \sum_{i \in F} \sum_{e' \in E \setminus M} {\textstyle \frac{\sum_{e \in N(e')} d(i, e)}{2}} \, x_{i, e'}
        \\&
        \quad \stackrel{\textrm{(III)}}{\leq} \sum_{i \in F} \sum_{e \in M} d(i, e) \, x_{i, e} + \sum_{i \in F} \sum_{e' \in E \setminus M} {\textstyle \big( d(i, e') + \frac{\sum_{e \in N(e')} d(e)}{2} \big)} \, x_{i, e'}
        \\&
        \quad = \sum_{i \in F} \sum_{e \in E} d(i, e) \, x_{i, e}
        + \sum_{e = \{j, k\} \in E} d(e) \sum_{i \in F} {\textstyle \frac{1}{2}} \Big[ \sum_{e' \in \delta(j) \setminus e} x_{i, e'} + \sum_{e' \in \delta(k) \setminus e} x_{i, e'}  \Big]
        \\&
        \quad \stackrel{\textrm{(IV)}}{\leq} \sum_{i \in F} \sum_{e \in E} d(i, e) \, x_{i, e}
        + \sum_{e \in M} d(e) \, \xtilde_e
        \\&
        \quad = \connflm(x) + \sum_{e \in M} d(e),
    \end{align*}
    where (I) follows from the fact that $\xtilde$ is fully supported by $M$, (II) from \eqref{eq:app:perfmat:reroute:id01}, (III) from \eqref{ineq:maxmat:reroute:costbound:dbound}, and (IV) again from \eqref{eq:app:perfmat:reroute:id01}.
\end{proof}

Then, the proof of Lemma~\ref{lem:app:perfmat:reroute} can be derived by Lemmas~\ref{lem:app:perfmat:reroute:feas01}, \ref{lem:app:perfmat:reroute:suppM}, and \ref{lem:app:perfmat:reroute:feas02}.
\end{document}